\documentclass[a4paper,11pt]{amsart}

\usepackage{amsthm}
\usepackage{amssymb}
\usepackage{amsmath}
\usepackage{mathtools}
\usepackage{pdfpages}
\usepackage{changes}
\usepackage{exscale,cite,color,amsopn}
\usepackage[colorlinks=true,pdftex,unicode=true,linktocpage,bookmarksopen,hypertexnames=false]{hyperref}
\usepackage[notref,notcite]{showkeys}
\usepackage{colortbl}
\usepackage{xypic}

\renewcommand{\leq}{\leqslant}
\renewcommand{\geq}{\geqslant}

\DeclareMathOperator{\id}{id}

\DeclareMathOperator{\Sym}{Sym}

\newcommand{\N}{\mathbb{N}}
\newcommand{\Z}{\mathbb{Z}}

\makeatletter
\numberwithin{equation}{section}
\numberwithin{figure}{section}
\numberwithin{table}{section}
\newtheorem{thm}{Theorem}[section]
\newtheorem*{thm*}{Theorem}
\newtheorem{lem}[thm]{Lemma}
\newtheorem{cor}[thm]{Corollary}
\newtheorem{pro}[thm]{Proposition}

\theoremstyle{definition} 
\newtheorem{defn}[thm]{Definition}

\newtheorem{question}[thm]{Question}

\newtheorem{exa}[thm]{Example}

\makeatother

\title[Decomposition theorems]{Decomposition theorems for involutive solutions to the
Yang--Baxter equation}
\author{S. Ram\'irez, L. Vendramin}

\address{IMAS--CONICET and Depto. de Matem\'atica, FCEN, Universidad de Buenos Aires, Pabell\'on~1,
Ciudad Universitaria, C1428EGA, Buenos Aires, Argentina}
\email{sramirez@dm.uba.ar}

\address{Department of Mathematics and Data
Science, Vrije Universiteit Brussel, Pleinlaan 2, 1050 Brussel}
\email{Leandro.Vendramin@vub.be}
\email{lvendramin@dm.uba.ar}

\keywords{}

\begin{document}

\begin{abstract}
Motivated by the proof of Rump of a conjecture of Gateva--Ivanova 
on the decomposability of square-free solutions to the Yang--Baxter equation, we present 
several other decomposability theorems based on the cycle structure of a 
certain permutation associated with the solution. 
\end{abstract}

\maketitle

\section{Introduction}

In \cite{MR1183474} Drinfeld considered 
set-theoretic solutions, i.e. pairs $(X,r)$, where $X$ is a non-empty set and 
$r\colon X\times X\to X\times X$ is a bijective map such that 
\[
	(r\times\id)(\id\times r)(r\times\id)=(\id\times r)(r\times\id)(\id\times r).
\]
The first systematic account
about these solutions 
was written by Etingof,
Schedler and Soloviev \cite{MR1722951} and Gateva-Ivanova and Van den Bergh 
\cite{MR1637256}. Both papers studied involutive solutions. The theory of such solutions 
was further developed in
several other papers, see for example 
\cite{MR2095675,MR2927367,MR2885602,MR2132760,MR2278047,MR2652212,MR3177933,MR2584610}.

Originally Gateva--Ivanova mainly studied involutive square-free 
solutions, i.e. solutions $(X,r)$ such that $r(x,x)=(x,x)$ holds for all $x\in X$. 
Square-free solutions produce algebraic structures with 
several interesting properties, see for example~\cite{MR2095675}. Nowadays it is clear 
that not only square-free solutions are of interest. In fact, several 
intriguing connections involving arbitrary involutive solutions were found recently~\cite{MR3291816}. 
These connections include 
radical rings~\cite{MR2278047}, homology~\cite{MR3558231}, 
ordered and diffuse groups~\cite{MR3572046,MR3815290,MR3974961} and Garside theory~\cite{MR2764830,MR3374524}. 

Recent research on the combinatorics of the Yang--Baxter equation is mainly 
based on the use of powerful algebraic methods. This leads to several breakthroughs in the area and 
new families of solutions were found. 
However, the problem of constructing new classes of solutions is still challenging 
and a program to classify all solutions is somewhat still out of reach. 
There are two promising approaches in this context. One is based on \emph{indecomposable} solutions 
and the other one 
is based on \emph{retractable} solutions. In both cases, 
the idea is to understand small solutions 
suitable to play the role of building blocks for the theory. 

The starting point of this work is an old conjecture of 
Gateva--Ivanova about the decomposability of certain solutions. In a 1996 talk at the \emph{International
Algebra Conference} in Hungary, she 
conjectured that finite \emph{square-free} solutions 
are always decomposable, that is solutions where the underlying set $X$ 
that can  be decomposed into a disjoint union $X=Y\cup Z$ of  
subsets $Y$ and $Z$ of $X$ such that $r(Y\times Y)\subseteq Y\times Y$ 
and $r(Z\times Z)\subseteq Z\times Z$. 

The conjecture was proved by Rump in~\cite{MR2132760}. In the same paper 
he also showed that the conjecture cannot be extended to infinite solutions. 

Based on the work of Etingof, Schedler and Soloviev~\cite{MR1722951}, 
we consider the diagonal of a solution $(X,r)$, a certain permutation $T$ of the set $X$. 
The cycle structure of this permutation turns out to be an invariant of the solution. 
Rump's theorem can then be restated as follows: If the diagonal
map of a solution $(X,r)$ fixes all points of $X$, then $(X,r)$ is decomposable.  

In this work we deal with decomposition theorems similar to that of Rump. 
The basic tools are the diagonal map of a solution and a recent result of 
Ced\'o, Jespers and Okni\'nski about solutions with a primitive permutation group~\cite{CJO}. 
It seems that Rump's result is 
just the tip of the iceberg, as --at least for finite solutions-- the bijective 
map $T$ turns out to encode deep combinatorial information about the structure of the solution. 
In particular, the map $T$ could be used to detect (in)decomposability or retractability of solutions.  

\medskip
The paper is organized as follows. We recall some preliminaries in Section~\ref{preliminaries}. 
We recall the classification of involutive solutions with a prime number of elements found by
Etingof, Schedler and Soloviev in~\cite{MR1722951}, Rump's theorem on 
decomposability of square-free solutions~\cite{MR2132760} and
the theorem of Ced\'o, Jespers and Okni\'nski on involutive 
solutions with primitive permutation groups~\cite{CJO}.  
Our results appear in Section~\ref{theorems}. 
In Theorem~\ref{thm:no_id} we prove that indecomposable solutions are always defined by 
permutations different from the identity. We remark that this result might 
be useful for the problem of constructing explicitly indecomposable solutions of 
small size. 
In Theorem~\ref{thm:n} we prove that
involutive solutions where the diagonal is a cycle of maximal length are indecomposable. 
Theorem~\ref{thm:fix1} deals with the decomposability of 
those solutions where the diagonal fixes only one point. 
Theorem~\ref{thm:n-2_n-3} is devoted to the decomposability of those solutions
where the diagonal fixes two or three points. We conclude the paper with some 
questions and conjectures. 

\section{Preliminaries}
\label{preliminaries}

A \emph{set-theoretic solution to the Yang--Baxter equation} is a pair 
$(X,r)$, where $X$ is a non-empty set and 
$r\colon X\times X\to X\times X$ is a bijective map such that 
\[
	(r\times\id)(\id\times r)(r\times\id)=(\id\times r)(r\times\id)(\id\times r).
\]
By convention, 
we write 
\[
r(x,y)=(\sigma_x(y),\tau_y(x)).
\]

A solution $(X,r)$ is said to be \emph{non-degenerate} if 
all the maps $\sigma_x\colon X\to X$ and $\tau_x\colon X\to X$ are bijective 
for all $x\in X$, and it is said to be \emph{involutive} if $r^2=\id_{X\times X}$. Note that 
for non-degenerate involutive solutions, 
\begin{equation*}
    \tau_y(x)=\sigma_{\sigma_x(y)}^{-1}(x)
    \quad\text{and}\quad
    \sigma_x(y)=\tau_{\tau_y(x)}^{-1}(y).
\end{equation*}

By convention, a \emph{solution} will be a non-degenerate involutive solution to the Yang--Baxter equation. 

A \emph{homomorphism} between the solutions $(X,r_X)$ and
$(Y,r_Y)$ is a map $f\colon X\to Y$ such that the diagram
 \[
 \xymatrix{
 X\times X
 \ar[d]_{f\times f}
 \ar[r]^-{r_X}
 & X\times X
 \ar[d]^{f\times f}
 \\
 Y\times Y
 \ar[r]^-{r_Y}
 & Y\times Y
 }
 \]
is commutative, i.e. $(f\times f)\circ r_X=r_Y\circ (f\times f)$. 
An isomorphism of solutions is a bijective homomorphism of solutions. 
Solutions form a category. 

A solution $(X,r)$ is said to be \emph{square-free}
if $r(x,x)=(x,x)$ for all $x\in X$. 

In~\cite[Proposition 2.2]{MR1722951} it is proved that if $(X,r)$ is a solution, then 
the map $T\colon X\to X$, $T(x)=\tau_x^{-1}(x)$, is invertible with inverse $x\mapsto\sigma_x^{-1}(x)$ 
and, moreover, 
\[
\tau_x^{-1}\circ T=T\circ\sigma_x
\]
for all $x\in X$. 

\begin{defn}
The \emph{diagonal} of a solution $(X,r)$ is defined as 
the permutation $T\colon X\mapsto X$, $x\mapsto\tau_x^{-1}(x)$.
\end{defn}

The \emph{permutation group} of a solution $(X,r)$ is defined as 
the subgroup $\mathcal{G}(X,r)$ of $\Sym_X$ generated by
the set $\{\sigma_x:x\in X\}$. Note that if $X$ is finite, then $\mathcal{G}(X,r)$ is finite. 
The group $\mathcal{G}(X,r)$ naturally acts on $X$. 
A solution $(X,r)$ is said to be \emph{indecomposable} 
if the group $\mathcal{G}(X,r)$ acts transitively 
on $X$ and \emph{decomposable} otherwise. This definition of decomposability 
turns out to be equivalent to the definition mentioned in the introduction, 
see for example the proof of~\cite[Proposition 2.12]{MR1722951}. 
Using the invertible map $T$ one proves that $\mathcal{G}(X,r)$ 
is isomorphic (as a permutation group on the set $X$) to the 
group generated by the set $\{\tau_x:x\in X\}$. 

\begin{thm}[Etingof--Schedler--Soloviev]
\label{thm:ESS}
Let $p$ be a prime number. 
An indecomposable solution with $p$ elements is
a \emph{cyclic solution}, i.e. a solution 
isomorphic to $(\Z/p,r)$, where
$r(x,y)=(y-1,x+1)$. 
\end{thm}

\begin{proof}
  See~\cite[Theorem 2.13]{MR1722951}
\end{proof}

We recall Rump's theorem. 

\begin{thm}[Rump]
\label{thm:Rump}
	Let $(X,r)$ be a finite solution. If the diagonal 
	of the solution is the identity, then $(X,r)$ is decomposable.
\end{thm}

\begin{proof}
    See~\cite[Theorem 1]{MR2132760}.
\end{proof}

A solution $(X,r)$ is said to be \emph{primitive} if the group $\mathcal{G}(X,r)$ acts primitively on $X$. 
In the Oberwolfach mini-workshop \emph{Algebraic Tools for 
Solving the Yang–Baxter Equation}, Ballester--Bolinches posed the following 
conjecture~\cite{MR4176795}: Finite primitive solutions are of size $p$ for some prime number $p$. 
The conjecture was proved by Ced\'o, Jespers and Okni\'nski. The proof uses 
the theory of braces and Theorem~\ref{thm:ESS}. 

\begin{thm}[Ced\'o--Jespers--Okni\'nski]
\label{thm:CJO}
	If $(X,r)$ is a finite primitive solution, then $|X|$ is a prime number and
	$(X,r)$ is a cyclic solution. 
\end{thm}

\begin{proof}
    See~\cite[Theorem 3.1]{CJO}. 
\end{proof}

In order to prove the conjecture mentioned in the introduction, 
Rump introduced a better way to deal with solutions. 
A \emph{cycle set} is a non-empty set $X$ provided with a 
binary operation $X\times X\to X$, $(x,y)\mapsto x\cdot y$, 
such that the maps $\varphi_x\colon X\to X$, $y\mapsto x\cdot y$, 
are bijective for all $x\in X$ and 
\[
(x\cdot y)\cdot (x\cdot z)=(y\cdot x)\cdot (y\cdot z)
\]
holds for all $x,y,z\in X$. 

A \emph{homomorphism} between the cycle sets $X$ and $Z$ is a map $f\colon
X\to Z$ such that $f(x\cdot y)=f(x)\cdot f(y)$ for all $x,y\in X$. An 
\emph{isomorphism} of cycle sets is a bijective homomorphism of cycle sets. 
Cycle sets form a category. 

A cycle set $X$ is said to be \emph{non-degenerate} if the map $X\to X$, $x\mapsto x\cdot x$, is bijective. 
Rump proved that solutions are in bijective correspondence with non-degenerate cycle sets. The correspondence
is given by
\[
	r(x,y)=((y*x)\cdot y,y*x),
\]
where $y*x=z$ if and only if $x=y\cdot z$, see \cite[Proposition 1]{MR2132760}. Moreover, 
the category of solutions and the category of non-degenerate cycle sets are equivalent. 

The permutation group $\mathcal{G}(X)$ of a cycle set $X$ is defined as the permutation group of its associated solution. 
Indecomposable and decomposble cycle sets are then defined in the usual way. 

The diagonal of a cycle set $X$ is then defined as the diagonal of its associated solution. 
Clearly, the cycle structure of the diagonal of a finite cycle set is invariant under isomorphisms. Note that
a solution is square-free if and only if the diagonal of its associated cycle set is the identity. 

We write $\Sym_X$ to denote the set of bijective maps $X\to X$. 

\begin{lem}
\label{lem:diagonal}
    Let $X$ be a cycle set with diagonal $T$. 
    If $U\colon X\to X$ is a bijective map and $U$ and $T$ are conjugate in $\Sym_X$, 
    then there exists an isomorphic cycle set structure
    on the set $X$ with diagonal $U$. 
\end{lem}

\begin{proof}
    Let $(x,y)\mapsto x\cdot y$ be the cycle set operation on $X$ and  
	let $\gamma\in\Sym_X$ be such that $U=\gamma^{-1}\circ T\circ \gamma$. A direct
	calculation shows that the operation $y\cdot_U z=\gamma^{-1}(\gamma(y)\cdot
	\gamma(z))$ turns $X$ into a cycle set isomorphic to the original $X$. Moreover, 
	\[
		y\cdot_U y=\gamma^{-1}(\gamma(y)\cdot\gamma(y))=\gamma^{-1}(T(\gamma(y)))=(\gamma^{-1}\circ T\circ\gamma)(y)=U(y)
	\]
	holds for all $y\in X$. 
\end{proof}


\section{Decomposition theorems}
\label{theorems}

For the first results in this section we use the language of cycle sets. We start
with some easy consequences of Theorem~\ref{thm:CJO} and the results of~\cite{MR3666217}. 

\begin{thm}\label{thm:decomposable}
Let $X$ be a finite cycle set with more than one element. If some $\varphi_x$ is a cycle of length at least $|X|/2$ and is coprime with $|X|$ then $X$ is decomposable.

\end{thm}

\begin{proof}
We first note that Theorem~\ref{thm:CJO} implies that in none of these cases $\mathcal{G}(X)$ can act primitively on $X$.
Let $y$ be a fixed point of $\varphi_x$ and $B$ a non-trivial block containing $y$. Since $\varphi_x$ fixes $y$ we must have 
    $\varphi_x(B)=B$. If this block were to contain an element that is not fixed by $\varphi_x$, and since this permutation is a cycle,
    it would have to contain all of them, which cannot happen since the block would then have more than $|X|/2$ elements. Since $y$ 
    was arbitrary this means any block containing a point fixed by $\varphi_x$ must consist entirely of fixed points. In particular 
    the size of the block must divide the number of fixed points. Since it also has to divide $|X|$ and these numbers are coprime, 
    there cannot be non-trivial blocks, contradicting the imprimitivity of the solution. 
\end{proof}

\begin{thm}
\label{thm:no_id}
    If $X$ is a finite cycle set 
    with more than one element 
    and there exists some $x\in X$ such that $\varphi_x=\id$, then 
    $X$ is decomposable. 
\end{thm}

\begin{proof}
    Let $x\in X$ be such that $\varphi_x=\id$. Then
    \[
    y\cdot z=(x\cdot y)\cdot (x\cdot z)=(y\cdot x)\cdot(y\cdot z)
    \]
    holds for all $y,z\in X$. This implies that 
    $\varphi_{y\cdot x}=\id$ for all $y\in X$ and hence 
    $\varphi_{g\cdot x}=\id$ for all $g\in \mathcal{G}(X,r)$. 
    Let $Y=\{x\in X:\varphi_x=\id\}$. If $Y=X$, then the cycle set
    is trivial and hence it is decomposable. If not, then the group $\mathcal{G}(X)$ acts on $Y$. Thus 
    $Y$ is a non-trivial $\mathcal{G}(X)$-orbit and hence $X$ is decomposable.
\end{proof}

Let us show some direct consequences of the theorem.
A cycle set $X$ is said to be \emph{retractable} if the equivalence relation on $X$ given by
\[
x\sim y\Longleftrightarrow\varphi_x=\varphi_y
\]
is non-trivial. This notion corresponds to that of \emph{retractable solutions}. 

\begin{cor}
\label{cor:regular}
Let $X$ be a finite cycle set such that $|X|>1$. If $\mathcal{G}(X)$ acts regularly on $X$, then $X$ is retractable.
\end{cor}

\begin{proof}
If $\mathcal{G}(X)$ acts regularly on $X$, then $|\mathcal{G}(X)|=|X|$. Since $X$ is indecomposable, 
$\varphi_x\ne\id$ for all $x\in X$ by Theorem~\ref{thm:no_id}. 
Then some $\varphi_x$ must appear at least twice as a defining permutation of $X$. 
\end{proof}

The following consequence goes back to Rump, see~\cite[Proposition 1]{MR2132760}.

\begin{cor}
If $X$ is an indecomposable finite cycle set with more than one element such that $\mathcal{G}(X)$ is abelian, then 
$X$ is retractable.
\end{cor}

\begin{proof}
The assumptions imply that $\mathcal{G}(X)$ acts regularly on $X$. Thus the claim follows from Corollary~\ref{cor:regular}.
\end{proof}

Now we present several results similar to that
of Rump about decomposability of square-free solutions. 
First we consider the case where the diagonal 
of the solution is a cycle of maximal length.

\begin{thm}
\label{thm:n}
	Let $X$ be a cycle set of size $n$. If $T$ is a cycle of length $n$, then $X$ is indecomposable.
\end{thm}

\begin{proof}
	Let $x,y\in X$. It is enough to prove that if $m\in\N$ and $T^m(x)=y$, then
	$x$ and $y$ belong to the same $\mathcal{G}(X)$-orbit.  This is clearly
	true if $m=1$. Now assume that the result holds for some $m\geq1$, then
	\[
	y=T^{m+1}(x)=T(T^m(x))
	\]
	implies that $T^m(x)$ and $y$ belong to the same
	$\mathcal{G}(X)$-orbit. Then the inductive hypothesis implies the claim. 
\end{proof}

In Theorem~\ref{thm:n}, 
it is important that the diagonal is a cycle of maximal length. A diagonal
that moves all points does not necessarily imply indecomposability. 

\begin{exa}
\label{exa:581}
	Let $X=\{1,\dots,6\}$ and $r(x,y)=(\sigma_x(y),\tau_y(x))$ be given by
	\begin{align*}
	    &\sigma_1=(123), && \sigma_2=(123)(456), &&\sigma_3=(123)(465), &&\sigma_4=\sigma_5=\sigma_6=(456),\\
	    &\tau_1=(132), && \tau_2=(132)(465), &&\tau_3= (132)(456),&& \tau_4=\tau_5=\tau_6=(465).
	\end{align*}
	Then $(X,r)$ is decomposable and $T=(123)(456)$ moves all
	points of $X$.  
\end{exa}

We shall need the following lemmas.

\begin{lem}
\label{lem:T}
Let $(X,r)$ be a solution and $x,y\in X$. Then $T(x)=y$ if and only if $r(y,x)=(y,x)$. 
\end{lem}

\begin{proof}
We first compute 
\[
r(T(x),x)=(\sigma_{T(x)}(x),\tau_xT(x))=(\sigma_{T(x)}(x),x).
\]
Since $(X,r)$ is involutive, it follows that $y=T(x)=\sigma_{T(x)}(x)$. 
Conversely, if $r(y,x)=(y,x)$, then $\tau_x(y)=x$ and hence $y=T(x)$. 
\end{proof}

The following result is well-known. 
It follows from a well-known general formula. 

\begin{lem}
\label{lem:Tx=x}
Let $(X,r)$ be a solution and $x\in X$ be such that $T(x)=x$. 
Then $T(\tau_y(x))=\tau_{\sigma_x(y)}(x)$ 
for all $y\in X$. 
\end{lem}

\begin{proof}
The Yang--Baxter equation applied to the tuple $(x,x,y)$ yields
\begin{align*}
    r_2(\sigma_x\sigma_x(y),\tau_{\sigma_x(y)}(x),\tau_y(x))=(\sigma_x\sigma_x(y),\tau_{\sigma_x(y)}(x),\tau_y(x)).
\end{align*}
Thus $r(\tau_{\sigma_x(y)}(x),\tau_y(x))=(\tau_{\sigma_x(y)}(x),\tau_y(x))$. 
Now Lemma~\ref{lem:T} implies that 
$T(\tau_y(x))=\tau_{\sigma_x(y)}(x)$. 
\end{proof}

Applying the lemma to square-free solutions we get that  
square-free cycle sets are balanced. This result was noted by Rump, 
see Definition 2 and the following 
paragraph of~\cite{MR2132760}. We provide
a proof for completeness. 

\begin{pro}
If $(X,r)$ is a square-free solution, then $\tau_x=\sigma_x^{-1}$ for all $x\in X$.
\end{pro}

\begin{proof}
  Let $x,y\in X$. We know that all solutions satisfy $\tau_x(y)=\sigma_{\sigma_y(x)}^{-1}(y)$. 
  Since $(X,r)$ is square-free, $T(y)=y$ and $T(\tau_x(y))=\tau_x(y)$. Lemma~\ref{lem:Tx=x} 
  implies that $\tau_{\sigma_y(x)}(y)=\tau_x(y)=\sigma_{\sigma_y(x)}^{-1}(y)$, from where the claim follows.
\end{proof}

We consider the case where the diagonal fixes one point. 

\begin{thm}
    \label{thm:fix1}
	Let $(X,r)$ be a finite solution of size $n>1$. If the diagonal of $(X,r)$ 
	is a cycle of length $n-1$,
	then $X$ is decomposable.
\end{thm}

\begin{proof}
  Let $x_0\in X$ be such that $T(x_0)=x_0$ and $Y=X\setminus\{x_0\}$. Assume that 
  $(X,r)$ is indecomposable. 
  Let $f\colon Y\to X$ be given by $y\mapsto \tau_y(x_0)$. 
  Since $(X,r)$ is indecomposable, there exists $y\in Y$ such that $\tau_y$ moves 
  $x_0$, as $x_0=T(x_0)=\tau_{x_0}(x_0)$ and $\mathcal{G}(X,r)=\langle\tau_y:y\in Y\rangle$. 
  So let $y\in Y$ be
  such that $f(y)\in Y$. 
  Since $T$ is a cycle of length $n-1$, 
  the permutation $T$ moves each element of $Y$. 
  This implies that
  \[
  Y=\{f(y),T(f(y)),\dots,T^{n-2}(f(y))\}.
  \]
  Since  
  \[
  T(f(y))=T(\tau_y(x_0))=\tau_{\sigma_{x_0}(y)}(x_0)=f(\sigma_{x_0}(y))
  \]
  by Lemma~\ref{lem:T}, it follows that
  \[
  Y=\{f(y),f(\sigma_{x_0}(y)),\dots,f(\sigma_{x_0}^{n-2}(y))\}.
  \]
  Then $\sigma_{x_0}$ permutes every element of $Y$ cyclically. 
  
  We claim that $\mathcal{G}(X,r)$ acts primitively on $X$. 
  If not, there exists a non-trivial block $B$. 
  Without loss of generality we may assume that $x_0\in B$. 
  Then there exists $y\in Y$ such that $y\in B$. Since $\sigma_{x_0}(x_0)=x_0$, it follows that
  $\sigma_{x_0}(B)=B$. In particular, $\sigma_{x_0}(y)\in B$ and hence $Y\subseteq B$ because
  $\sigma_{x_0}$ is a cyclic permutation on $Y$, a contradiction. 
  
  Since $\mathcal{G}(X,r)$ acts primitively on $X$, Theorem~\ref{thm:CJO} implies that
  $(X,r)$ is a cyclic solution. This is a contradiction, as the diagonal 
  of a cyclic solution has no fixed points.
\end{proof}

By inspection, solutions of size $\leq 10$ 
such that the diagonal has only one fixed point are decomposable. This is not necessarily 
true if the diagonal fixes more than one point. 

\begin{exa}
    Let $X=\{1,2,\dots,8\}$ and $r(x,y)=(\sigma_x(y),\tau_y(x))$, where
    \begin{align*}
    \sigma_1&=(45), &\sigma_2&=(36), 
    &\sigma_3&=(27), &\sigma_4&=(18),\\
    \sigma_5&=(13428657), &\sigma_6&=(17568243), 
    &\sigma_7&=(12468753), &\sigma_8&=(13578642),\\
    \tau_1&=(46), &\tau_2 &= (35),
    &\tau_3&=(28), &\tau_4 &= (17),\\
    \tau_5&=(18657243), &\tau_6&=(13427568), 
    &\tau_7&=(13687542), &\tau_8&=(12457863).
    \end{align*}
    Then $T=(57)(68)$ has four fixed points and $(X,r)$ is indecomposable. 
\end{exa}

\begin{exa}
    Let $X=\{1,2,\dots,9\}$ and $r(x,y)=(\sigma_x(y),\tau_y(x))$, where
    \begin{align*}
    \sigma_1 &= (167925483), & \sigma_2 &=(125983467), &\sigma_3&=(165923487),\\
    \sigma_4 &= (158936472), & \sigma_5 &=(136972458), & \sigma_6 &=(156932478),\\ 
    \sigma_7 &= (149)(268), & \sigma_8 &=(149)(357), & \sigma_9 &=(268)(357),\\
    \tau_1 &= (185639742), & \tau_2 &=(189632745), &\tau_3&=(139642785),\\
    \tau_4 &= (198623754), & \tau_5 &=(128653794), & \tau_6 &=(194628753),\\ 
    \tau_7 &= (295)(384), & \tau_8 &=(176)(295), & \tau_9 &=(176)(384).
    \end{align*}
    Then $T=(123)(456)$ has three fixed points and $(X,r)$ is indecomposable. 
\end{exa}


Now we consider the case where the diagonal fixes two or three points. 

\begin{thm}
\label{thm:n-2_n-3}
Let $(X,r)$ be a finite solution of size $n$. 
\begin{enumerate}
    \item If $T$ is an $(n-2)$-cycle and $n$ is odd, then $X$ is decomposable. 
    \item If $T$ is an $(n-3)$-cycle and $\gcd(3,n)=1$, then $X$ is decomposable.
\end{enumerate}
\end{thm}

\begin{proof}
  Let $m$ be the length of $T$, 
  \begin{align*}
      Y=\{y\in X:T(y)\ne y\} &&\text{and} && Z=\{z\in X:T(z)=z\}. 
  \end{align*}
  If $X=Y\cup Z$ is a decomposition of $X$, there is nothing to prove. If not, 
  there exists $z\in Z$ and $x\in X$ such that $\tau_x(z)\in Y$. Since $T$ permutes the elements of $Y$ cyclically, 
  the elements $\tau_x(z),T(\tau_x(z)),\dots,T^{m-1}(\tau_x(z))$ are all distinct. Moreover, since $z$ is 
  a fixed point, 
  \[
  \{\tau_x(z),T(\tau_x(z)),\dots,T^{m-1}(\tau_x(z))\}=
  \{\tau_x(z),\tau_{\sigma_z(x)}(z),\dots,\tau_{\sigma_z^{m-1}(x)}(z)\}. 
  \]
  In particular, $x,\sigma_z(x),\dots,\sigma_z^{m-1}(x)$ are all distinct elements, so the decomposition of 
  $\sigma_z$ into disjoint cycles contains a cycle of length at least $m$. Since $T$ is a cycle of length $m$, 
  $T^{k+m}(\tau_x(z))=T^k(\tau_x(z))$, i.e. 
  \[
  \tau_{\sigma_z^{k+m}}(x)(z)=\tau_{\sigma_z^k(x)}(z).
  \]
  Moreover, $m$ is the minimal period which implies that the length of this cycle in $\sigma_z$ has to be divisible by $m$. 
  By the hypothesis on the $m$ we are taking this can only happen if the length of the cycle is exactly $m$.

  In the case where $m=n-2$, since $\sigma_z$ has at least one fixed point $z$, $\sigma_z$ has to be a cycle of length $m$.
  By Theorem \ref{thm:decomposable} $X$ must be decomposable. In the case that $m=n-3$ there are two possibilities, either 
  $\sigma_z$ is a cycle of length $n-3$, and we can conclude the same way as before, or it is a product of cycle of length 
  $n-3$ and a cycle of length 2. If it is a product of two cycles we take a non-trivial block containing $z$. Like 
  before this block cannot contain an element of the cycle of length $n-3$. If it does not but it contains an element 
  that belongs to the cycle of length 2 then it also has to contain the other element of the cycle. This means the block 
  has size 3 and this is not possible because $\gcd(3,n)=1$.
\end{proof}

The assumption on $n$ cannot be removed from the first claim of Theorem~\ref{thm:n-2_n-3}.

\begin{exa}
\label{exa:T=12}
Let $X=\{1,2,3,4\}$ and $r(x,y)=(\sigma_x(y),\tau_y(x))$, where
\begin{align*}
\sigma_1&=(34), &\sigma_2&=(1324), & \sigma_3&=(1423), & \sigma_4&=(12),\\
\tau_1&=(24),  &\tau_2&=(1432), & \tau_3&=(1234), & \tau_4&=(13).
\end{align*}
Then $T=(12)$ and $(X,r)$ is indecomposable. 
\end{exa}

Indecomposable solutions of small size are known up to size 11 and there
are only 172 such solutions, see
Table~\ref{tab:indecomposable}~\cite{AMV}. 

\begin{table}[h]
\caption{Indecomposable solutions of size $n\leq11$.}
\begin{tabular}{|r|cccccccccc|}
\hline
$n$ & 2 & 3 & 4 & 5 & 6 & 7 & 8 & 9 & 10 & 1\tabularnewline
\hline 
number & 1 & 1 & 5 & 1 & 10 & 1 & 100 & 16 & 36 & 1\tabularnewline
\hline
\end{tabular}
\label{tab:indecomposable}
\end{table}

\begin{question}
Is it true that 
\[
\lim_{n\to\infty}\frac{\text{\# decomposable solutions of size $n$}}{\text{\# solutions of size $n$}}=1?
\]
\end{question}

The number of known indecomposable solutions is quite small 
for making reasonable conjectures. However, the following questions seem to be
interesting.

\begin{question}
    Let $(X,r)$ be a finite solution. Is it true that if some $\sigma_x$ contains a cycle of length coprime with $|X|$ then $(X,r)$ is decomposable?
\end{question}


\begin{question}
Let $(X,r)$ be a finite solution. How many fixed points of the diagonal of 
$(X,r)$ guarantee the decomposability of $(X,r)$?
\end{question}

Computer calculations suggest that if the diagonal of a solution $(X,r)$ 
fixes more than $|X|/2$ points, then $(X,r)$ is decomposable. 
We do not know, for example, 
what happens if the diagonal is a transposition or a 3-cycle. 

\begin{question}
Are there indecomposable solutions, where the diagonal is a transposition, 
different from that of Example~\ref{exa:T=12}?
\end{question}

\begin{question}
Are there indecomposable solutions 
where the diagonal is a 3-cycle?
\end{question}

In Table~\ref{tab:T} we show
the cycle structure of the diagonal of small indecomposable solutions. 

\begin{table}[h]
\caption{Cycle structure of the diagonal of small indecomposable solutions.}
\begin{tabular}{|c|c|c|} 
\hline
 $n$ & number & cycle structure\\ 
 \hline
 4 & 3 & 4 \\ 
   & 2 & 2 \\ 
 \hline
 6 & 6 & 6\\
   & 3 & 3,3\\
   & 1 & 2,2,2\\
  \hline
 8 & 7 & 2,2\\
 & 25 & 2,2,2,2\\
 & 2 & 2,2,4\\
 & 1 & 2,6\\
 & 14 & 8\\
 & 51 & 4,4\\
 \hline
 9 & 1 & 3,3\\
 & 6 & 3,3,3\\
 & 9 & 9\\
 \hline
 10 & 1 & 2,2,2,2,2\\
 & 15 & 5,5\\
 & 20 & 10\\
 \hline
\end{tabular}
\label{tab:T}
\end{table}





\section*{Acknowledgments}

Vendramin acknowledges the support of
NYU-ECNU Institute of Mathematical Sciences at NYU Shanghai. 
This work is supported 
by PICT 2016-2481 and UBACyT 20020170100256BA. 

\bibliographystyle{abbrv}
\bibliography{refs}

\end{document}